\documentclass[a4paper,UKenglish]{article}
 
\usepackage{microtype}%if unwanted, comment out or use option "draft"
\usepackage{fullpage}
\RequirePackage{subcaption} % subfig
\title{Two-sided Quantum Amplitude Amplification and Exact-Error Algorithms}

\author{Debajyoti Bera \thanks{IIIT-Delhi, New Delhi, India. Email: {\tt
dbera@iiitd.ac.in}}}

\usepackage[pdftex]{graphicx,color}

\usepackage{amsmath,amssymb,amsthm,thmtools,thm-restate}
\usepackage{hyperref}
\usepackage{xspace}

\newcommand{\ket}[1]{|#1\rangle}
\newcommand{\bra}[1]{\langle #1 |}
\newcommand{\xor}{\oplus}
\def\tensor{\otimes}
\newcommand{\braket}[2]{\langle #1 | #2 \rangle}
\newcommand{\ketbra}[2]{\ket{#1}\bra{#2}}

\newcommand{\C}{\mathcal{C}}
\renewcommand{\H}{\mathcal{H}}
\newcommand{\I}{\imath}

\newtheorem{theorem}{Theorem}
\newtheorem{lemma}{Lemma}
\newtheorem{definition}{Definition}
\newtheorem{corollary}{Corollary}

\newcommand{\RQP}{\ensuremath{\mathbf{RQP}}\xspace}
\newcommand{\EQP}{\ensuremath{\mathbf{EQP}}\xspace}
\renewcommand{\P}{\ensuremath{\mathbf{P}}\xspace}

\newcommand{\BPP}{\ensuremath{\mathbf{BPP}}\xspace}

\newcommand{\RP}{\ensuremath{\mathbf{RP}}\xspace}
\newcommand{\BQP}{\ensuremath{\mathbf{BQP}}\xspace}
\newcommand{\ERP}{\ensuremath{\mathbf{ERP}}\xspace}
\newcommand{\ERQP}{\ensuremath{\mathbf{ERQP}}\xspace}
\newcommand{\EBPP}{\ensuremath{\mathbf{EBPP}}\xspace}
\newcommand{\EBQP}{\ensuremath{\mathbf{EBQP}}\xspace}

\newcommand{\pair}[2]{\langle #1, #2 \rangle}
\newcommand{\A}{\mathcal{A}}
\newcommand{\BB}{\mathcal{B}}

\newcommand{\Q}{\mathcal{Q}}
\newcommand{\Proj}{\mathcal{P}}
\newcommand{\llangle}{\big\langle}
\newcommand{\rrangle}{\big\rangle}
\renewcommand{\S}{\mathcal{S}}
\newcommand{\half}{\frac{1}{2}}
\newcommand{\quarter}{\frac{1}{4}}

\begin{document}
\maketitle

\begin{abstract}
Amplitude amplification is a central tool used in Grover's quantum search algorithm and has been used in various forms in numerous quantum algorithms since then. It has been shown to completely eliminate one-sided error of quantum search algorithms where input is accessed in the form of black-box queries. We generalize amplitude amplification for two-sided error quantum algorithm for decision problems in the familiar form where input is accessed in the form of initial states of quantum circuits and where arbitrary projective measurements may be used to ascertain success or failure.  This generalization allows us to derive interesting applications of amplitude amplification for distinguishing between two given distributions based on their samples, detection of faults in quantum circuits and eliminating error of one and two-sided quantum algorithms with exact errors.
%We were able to further show that the complexity class $\EQP$ (quantum analog
%of $\P$) is exactly same as $\ERQP$ and $\EBQP$. The latter classes are the
%``exact'' versions of quantum bounded-error complexity classes in which the
%probability of error for ``yes'' and ``no'' instances had to be exactly equal to
%two known threholds, unlike $\RQP$ and $\BQP$ (quantum analogs of $\RP$ and
%$\BPP$) in which those probabilities were only required to be upper or lower
%bounded by some threshold.
\end{abstract}

\section{Introduction}\label{section:intro}
The motivation behind this work is to investigate the characteristics of quantum
computation when viewed as randomized algorithms. It is known that quantum
amplitude amplification, the key technique underlying Grover's unordered
search algorithm, is able to reduce and even eliminate error of one-sided quantum
black-box algorithms for search problems~\cite{BHMT}. We explored that direction further for
two-sided error algorithms for decision problems based on
the key observation that quantum algorithms appear to be
better at distinguishing between two given probability distributions compared to
classical randomized algorithms.

Suppose we are given a biased coin whose
distribution is either
$\mu_1 = \langle 1/3, 2/3 \rangle$ or $\mu_2 = \langle 2/3, 1/3 \rangle$.
A classical problem of probabilistic classification is to determine the
distribution of the coin by tossing it several times.
Various techniques exist like Bayesian classification and
maximum likelihood estimation, all of which aim to minimize some kind of error
that is inherent in such a probabilistic inference. But it is not believed to be
possible to confidently classify a distribution without any error. This is
true even if $\mu_1 = \langle 0,1\rangle$ instead.

However, such classification is possible when the distributions come
from a {\em quantum system}, our definition of a quantum source of
random samples.
We define a quantum system (QS) as a combination of a quantum circuit
$C$, an input to the circuit $\ket{\psi}$ and
a two-outcome projective measurement operator $\Proj=\langle P_E, I -
P_E \rangle$ (two outcomes will be {\em always labeled as $E$ and $F$} for
convenience) and denote it by $\langle \ket{\psi}, C, \Proj \rangle$. If we are given an actual
instance of a QS and we {\em
apply the circuit on the input followed by measurement using the projective
operator}, we will obtain a
sample in $\{E,F\}$ from the output probability distribution $\langle p_E, 1-p_E
\rangle$ where $p_E$
denotes the probability of observing outcome $E$ when
$C\ket{\psi}$ is measured using $\Proj$.
%Henceforth, we will identify any quantum
%system by its output probability distribution.

The quantum version of the above question of classifying between $\mu_1$ and
$\mu_2$ becomes this:
given an
instance $\Q$ which can be either a quantum system $\Q_1$ with output distribution $\mu_1$ or
QS $\Q_2$ with output distribution $\mu_2$, can we confidently
figure out if $\Q$ is $\Q_1$ or $\Q_2$ (in other words, determine the actual distribution
of $\Q$) by using $\Q$ in a black-box manner? Assume that both $\Q_1$ and$\Q_2$ involve the same
number of qubits and the same set of outcomes ($E$ and $F$). This is analogous to asking if two or more distinct
distributions (over same support of two elements) can be distinguished
without any probability of error. Even though classical techniques cannot
identify the exact distribution from the sample distribution without any
error, 
we show that it is possible to do so for
distributions of quantum systems.

\begin{theorem}\label{theorem:application}
    Given a quantum system $\Q = \langle \ket{\psi}, C, \Proj \rangle$ whose output distribution can either be
    $\pair{\delta}{1-\delta}$ or $\pair{\epsilon}{1-\epsilon}$ for some $0 \le \delta
	< \epsilon \le 1$, there is a quantum circuit $C'$ which can determine
    the output distribution of $\Q$ without any probability of error. $C'$ takes
    $\ket{\psi}$ as input, makes repeated calls to $C$, $C^\dagger$ and
    employs gates that depend upon operators of $\Proj$ and $\ket{\psi}$.
\end{theorem}

The core technique is once again, {\em quantum amplitude amplification}.
It can be thought of as a quantum analog of repeated
trials used in randomized algorithms for reducing mis-classification error.
It is the workhorse
behind Grover's famous quantum unordered search algorithm~\cite{Grover1996} and was later
shown to be also applicable to the Deutsch-Jozsa
problem~\cite{BeraDJ2015}.
It appears that quantum algorithm designers simply
cannot wave it enough; it is applicable to almost any search
problem to yield a surprising improvement, usually quadratic, over classical algorithms. 
Since its inception, amplitude amplification have been used, either directly
or in the form of Grover's search algorithm for a vast range of problems like minimum of an unordered
array~\cite{Durr96}, minimum spanning tree~\cite{Durr2004} and even
clustering~\cite{Aimeur2007}.
Nevertheless, we feel that the technique still has a long way to go, especially,
when used in a non-blackbox manner.

%
%Suppose we are given an arbitrary pure
%quantum state in a state space $\H$ with a certain probability of being observed in a particular
%``good'' subspace of $\H$. This technique describes how to create operators
%which can be applied on the state (maybe multiple times) to ``amplify'' the
%``good'' probability. A simpler, essentially classical, process would be take
%multiple measurements on multiple such states and make a statistical inferenc
%from them. However, the benefit of amplitude amplification is that the number of
%times the operator has to be applied is asymptotically square root of the number
%of statistic samples a classical process may need.
%
%The first technical result was proved in Grover's seminal
%paper~\cite{Grover1996}
%in which good probability $1/N$ ($N$ stands for the
%dimension of the state space) was amplified to ``nearly 1'' in $O(\sqrt{N})$
%steps. This 
The most generalized and popular version of this technique was given by
Brassard et al.
\begin{theorem}[Exact amplitude amplification~\cite{BHMT}]\label{theorem:bhmt}
    Consider a Boolean function $\Phi: X \to \{0,1\}$ that partitions a set $X$
    between its {\em good} (those which $\Phi$ evaluates to 1) and {\em bad} (those
    which evaluate to 0) elements. Consider also a quantum algorithm that uses
    no measurements and uses oracle gates for computing $\Phi$
    such that $C\ket{0}$ is quantum superposition of the elements of $X$ and let
    $a > 0$ denote the success probability that a good element is observed if
    $C\ket{0}$ is measured (in the standard basis).
    There exists a quantum circuit (that depends upon $a$) which finds a
    good solution with certainty using at most $\Theta(1/\sqrt{a})$ applications of $C$ and
    $C^\dagger$.
\end{theorem}

This theorem is highly versatile as it is. However, for our applications we
require further generalizations. For example, we are interested in not only
one-sided, but also two-sided error algorithms. We also want to apply it to
algorithms which are measured not necessarily in the standard basis. Lastly,
we want algorithms which act on non-$\ket{0}$ input states, specifically, input
states that correspond to the input $\Phi$, suitably encoded -- this is similar to classical Boolean circuits without
oracle gates.
Lastly, for the results of this paper we stick to only
decision versions of the above theorem (though our results could be extended to
circuits that output some solution).
The following theorem is our version of Theorem~\ref{theorem:bhmt} with the constraint that the probability $a$ is fixed for every
possible $\Phi$ (condition of {\em exactness}).

\begin{theorem}[Decision version of generalized exact amplitude
    amplification]\label{theorem:aa-two-sided}
Consider a Boolean function $\Phi: X \to \{0,1\}$ that partitions a set $X$
between its {\em good} (those which $\Phi$ evaluates to 1) and {\em bad} (the
rest of $X$) elements.
Suppose $C$ is a quantum algorithm (or circuit) that uses no measurement and
decides $\Phi$ with two-sided exact error $(\delta,\epsilon)$ for some $\delta <
\epsilon$. That is, the probability of error when $C$ is given a good $x \in
X$ is
{\em exactly} $\epsilon$ and when $x$ is bad is {\em exactly} $\delta$.
Here success and error is determined upon measurement of the output state of
$C$ by any projective measurement with two outcomes.
There exists a quantum circuit $C'$ that
calls $C$ and $C^\dagger$, uses the same input as that of $C$ (maybe with 
ancill\ae), is measured using an
extension of the measurement operator for $C$ and decides $\Phi$ with certainty,
\end{theorem}

The primary contribution of this paper are a few interesting applications of
amplitude amplification. If we have two quantum
systems which differ only in their circuit, then we can essentially use their
output distribution, after suitably amplifying the systems, to distinguish
between those circuits. We show how this can be used to detect faults in
quantum circuits.

On the other hand, if we have two systems that differ only in
their input states, then we get a way to amplify their probability of
acceptance. This is exactly at the core of our proof that quantum classes
equivalent to exact two-sided and exact one-sided error classes can be
``derandomized'', in the sense that their errors can be completely eliminated.

One of the major, and still open, questions of {\em Complexity
Theory} is how \P compares to \RP and
\BPP, one-sided and two-sided bounded error polynomial-time complexity classes. The current best results
are the obvious inclusions $\P \subseteq \RP \subseteq \BPP$, though there are
some evidences of their equivalence.
Same question for their quantum analogs is in an equally indeterminate state,
%The quantum analogs of the above questions are in an equally indeterminate state,
i.e., $\EQP \subseteq \RQP \subseteq \BQP$; these are quantum analogs of \P,
$\RP$ and $\BPP$, respectively. There is not even much evidence
that $\EQP = \BQP$. One
approach towards settling this question is studying restricted versions of these
classes. Our results show that their
exact error versions, $\ERQP$ and $\EBQP$, are identical to $\EQP$ as long as
the two(one)-sided errors are fixed
for all instances~\footnote{The same question for classical classes was asked here:
\url{http://cstheory.stackexchange.com/questions/20027/in-what-class-are-randomized-algorithms-that-err-with-exactly-25-chance}.}.

\vspace*{1em}\noindent{\bf Organization:} The rest of the paper is organized as
follows.
We discuss quantum distinguishability of quantum systems in
Section~\ref{section:distinguish-qs}. The proof of our main theorem on
distinguishability is given in Section~\ref{section:proof}. This theorem, even
though quite general, is not suitable enough to amplify a collection of
quantum systems in a uniform manner; in Section~\ref{section:uniform} we
discuss a uniform version of our main theorem.
Section~\ref{section:distinguish-circuits} contains one of the
applications about detection of faults in quantum circuits and
in Section~\ref{section:exact-error-classes} we show that $\EBQP=\ERQP=\EQP$ and
prove Theorem~\ref{theorem:aa-two-sided} for regular circuits and those with
oracle gates.

\section{Distinguishing quantum systems}\label{section:distinguish-qs}
We will use $\mu_p$ to denote a distribution $\pair{p}{1-p}$ over outcomes
$\pair{E}{F}$ and $\mu(\Q)$ to denote output distribution of a quantum system $\Q$.

As explained earlier, the main problem we are interested in involves a given instance of a quantum system $\Q$ which
can be either $\Q_\delta$ with output
distribution $\mu_\delta = \pair{\delta}{1-\delta}$ or $\Q_\epsilon$ with output
distribution $\mu_\epsilon = \pair{\epsilon}{1-\epsilon}$. We want to construct
a quantum algorithm, rather a circuit, that can ``call $\Q$ as a subroutine''
and determine if $\Q=\Q_\delta$ or $\Q=\Q_\epsilon$.

We can even extend this
to multiple quantum systems $\S = \{\Q_1, \Q_2, \ldots \}$ where output
distribution of any $\Q_i$ is either $\mu_\delta$ or $\mu_\epsilon$. We use the
notation $QD(\Q_1, \Q_2, \ldots)$ or even shorter $QD(\S)$ to refer to the {\em
quantum distinguishability} problem among quantum systems of $\S$.

Our goal is to design a quantum circuit in which we can ``embed any given $\Q$'' as a
black-box.
This motivated us to define a notion of black-box extension for quantum
systems, similar to quantum algorithms with subroutines or quantum circuits with
black-box operators, allowing only trivial extensions to inputs states and
projection
operators. We refer to these as $\BB$-transforms ($\BB$
standing for ``black-box''). A general illustration is given in Figure~\ref{fig:b_transform}.

\begin{figure*}[!h]
\centering \resizebox{0.8\linewidth}{!}{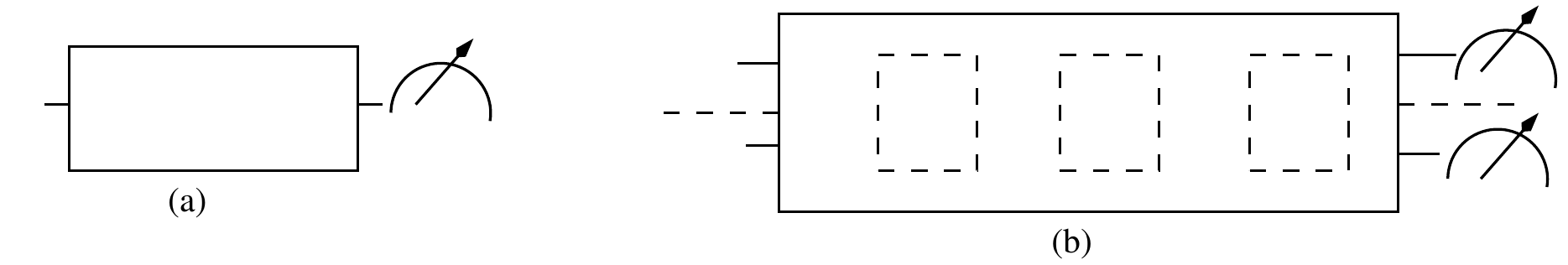} %\qquad %
\caption{Schematic for $\BB$-transform}
\label{fig:b_transform}
\end{figure*}

\begin{definition}[$\BB$-transform]\label{defn:bb-transform}
    A (non-uniform) $\BB^n_{\delta,\epsilon}$-transform for $n$-qubit systems is
    a (non-uniform) procedure for extending an $n$-qubit QS $\Q_1
    = \llangle \ket{\psi}, C, \Proj \rrangle$ to a (possibly larger) QS $\Q_2 =  \llangle \ket{\psi'}, C', \Proj' \rrangle$ 
    whose components are black-box extensions of
    the components of $\Q_1$.% in the following manner.
    \begin{itemize}
	\item The input in $\Q_2$ is an extension of the input in $\Q_1$
	    supplemented by ancill\ae\ qubits initialized to a fixed
	    state (wlog. in state $\ket{\mathbf{0}}$), i.e., 
	    $\ket{\psi'}=\ket{\psi} \otimes \ket{00\cdots 00}$.% The number of
	    %ancill\ae\ is independent of $\Q_1$ and depend upon $\delta, \epsilon$.
	\item The projection operator in $\Q_2$ is an extension of the projection
	    operator in $\Q_1$ to include measurement of the ancill\ae\ in a basis
	    independent of $\Q_1$, i.e., $\Proj' = \Proj \otimes \Proj_a$.
	    %$\Proj_a$ is a projective measurement operator for the ancill\ae\ 
	    %qubits that depend upon $\delta, \epsilon$.
	\item The number of ancill\ae\ and the operator $\Proj_a$ are independent of $\Q_1$ and depend upon $\delta, \epsilon$.
	\item The circuit in $\Q_2$ calls $C$ and $C^\dagger$ and uses
	    additional gates that depend upon $\delta$ and $\epsilon$.
	\item $C'$ may also use gates that depend upon $\Proj_E$ and
	    $\ket{\psi}$.
    \end{itemize}
\end{definition}

We call the transformations that satisfy the final condition as ``non-uniform'' since the transformed circuit
could be using gates that depend upon the input states and measurement operators
of the respective quantum system.
Note that the non-uniformity is not with respect to $n$, the number of qubits of
the quantum system, but with respect to the gates of the transformed
circuit. It will be clear from the proof of
Theorem~\ref{theorem:main-separable} that the
transformations that will be used in this paper are anyway uniform in $n$. In any
case, we will always drop $n$ from the superscript of $\BB^n_{\delta,\epsilon}$.
We will revisit the notion of non-uniformity in Section~\ref{section:uniform}. 

We want transformed quantum circuits that solve the quantum distinguishing problem without any error which
motives the next definition.

\begin{definition}\label{defn:solution}
    For a set of quantum systems $\S = \{\Q_1, \Q_2, \ldots \}$ with output distributions
    either $\mu_p$ or $\mu_q$ (for $p < q$), a
    $\BB$-transform $\BB$ is said to solve $QD(\S)$ with error
    $(\delta,\epsilon)$, in other words $\BB$ is a
    $(\delta,\epsilon)$-solution of $QD(\S)$, if the following holds for some
    $\delta < \epsilon$ and all $\Q
    \in \S$.
    \begin{itemize}
	\item If $\mu(\Q) = \mu_p$, then outcome of $\BB(\Q)$ is $E$ with
	    probability $\delta$.
	\item If $\mu(\Q) = \mu_q$, then outcome of $\BB(\Q)$ is $E$ with
	    probability $\epsilon$.
    \end{itemize}
    $QD(\S)$ is said to have a {\em perfect solution} if $\BB$ is a
    $(0,1)$-solution of $QD(\S)$.
\end{definition}

It can be seen that the identity $\BB$-transform is a trivial solution of the
above $QD(\S)$ with error $(p,q)$.
The last part of the above definition is based on the fact that if $\BB$ is a
$(0,1)$-solution of $QD(\S)$, then the outcome of $\Q' = \BB(\Q)$ can be used to
correctly infer the output distribution of any given instance $\Q \in \S$.
Let $\Q' = \BB(\Q)$ -- which is essentially an extension of the input of $\Q$
with some ancill\ae, an extension of its measurement operator and a circuit that
can call the circuits of $\Q$ (and its inverse) in a black-box manner. 
If the output distribution of $\Q$ is $\mu_p$, then the outcome of
$\Q'$ is never $E$ and otherwise (i.e., if the output
distribution of $\Q$ is $\mu_q$) the outcome of $\Q'$ is always $E$ without any error.

The main theorem of our work is stated next.

\begin{theorem}\label{theorem:main-separable}
    Let $\S = \{\Q_1, \Q_2, \ldots \}$ be a collection of quantum systems such
    that output distribution of any $\Q_i \in \S$ is either $\mu_\delta$
    or $\mu_\epsilon$ for some $\delta < \epsilon$.
    Then $\S$ is perfectly-solvable via some $\BB$-transition
    $\BB_{\delta,\epsilon}$, i.e., any $\Q_i \in \S$ can be transformed by $\BB_{\delta,\epsilon}$ to some
    $\Q'_i$ such that:
    \begin{itemize}
	\item if output distribution of $\Q_i$ is $\mu_\delta$, then outcome of
	    $\Q'_i$ is never $E$ and
	\item if output distribution of $\Q_i$ is $\mu_\epsilon$, then outcome of
	    $\Q'_i$ is always $E$.
    \end{itemize}
\end{theorem}

The proof of this theorem is presented in the next section.
Note that, unlike Theorem~\ref{theorem:bhmt} which only applies to one-sided error
algorithms, we prove that two-sided error algorithms can also be ``amplified to
certainty''.
A straight-forward application of
this is to exactly distinguish between two QS with known output
distributions, such as Theorem~\ref{theorem:application}
(Section~\ref{section:intro}).

\begin{proof}[Proof of Theorem~\ref{theorem:application}]
Consider the transformation $\BB_{\delta,\epsilon}^n$ from
Theorem~\ref{theorem:main-separable}. Given an $n$-qubit $\Q = \langle \ket{\psi}, C, \Proj
\rangle$, construct the transformed QS $\BB^n_{\delta,\epsilon}(\Q)=\big\langle \ket{\psi}\otimes \ket{\mathbf{00\ldots
0}}, C', \Proj \otimes \Proj_a \big\rangle$.
By Theorem~\ref{theorem:main-separable}, the output state of the transformed
circuit $C'$, when given
$\ket{\psi}$ (along with a few ancill\ae\ in a fixed state), upon measurement by
a simple extension of $\Proj$, has outcome either $E$ or $F$, depending upon
whether $\mu(\Q)=\mu_\delta$ or $\mu(\Q)=\mu_\epsilon$.
\end{proof}

\section{Proof of Theorem~\ref{theorem:main-separable}}
\label{section:proof}
We first state and prove our main technical tool -- the {\em Separability
Lemma} which essentially amplifies amplitudes of one-sided error algorithms. The
Lemma can be proven using already known techniques of amplitude amplifications
(e.g., see \cite[Sec 2.1]{BHMT}).
We give an alternative recursive construction that is optimized
towards amplifying fixed probabilities.

We use the following notation for the sake of brevity. Given a collection of
quantum systems $\{\Q_1, \Q_2, \ldots \}$ (such collections will be always
denoted by $\S$), we say that $\S$ is
$(\delta,\epsilon)$-separable (for some $\delta <\epsilon$) if output
distribution of any $\Q_i$ in $\S$ is either $\mu_\delta$ or $\mu_\epsilon$.

\begin{lemma}\label{lemma:fully-separable}[Separability]
    For $\delta < \epsilon < 1$ and a collection of quantum systems $\S_1$ which
    is $(\delta,\epsilon)$-separable, there is a $\BB$-transform $\BB_\epsilon$ which converts
    $\S_1$ to a $(\delta',1)$-separable collection of quantum systems (for some $\delta \le \delta' < 1$). Additionally, $\delta=\delta'=0$ if and only if $\delta=0$.
    %Furthermore, if all the input states of $\S_1$ are same or form an
    %orthonormal set, then $\BB_\epsilon$ is uniform.
\end{lemma}

Given an instance $\Q = \langle \ket{\psi}, C, \Proj \rangle$ of some $\Q_i \in \S_1$, Lemma~\ref{lemma:fully-separable} gives us a way to
determine whether the distribution of $\Q$ is $\pair{0}{1}$ or
$\pair{\epsilon}{1-\epsilon}$ by first transforming $\Q$ to $\BB(\Q) = \Q' =
\langle \ket{\psi'}, C', \Proj' \rangle$ and then measuring the output of $C'$
on $\ket{\psi'}$ (which is a simple extension of the original input state)
using measurement operator $\Proj'$ (which is also a simple extension of the
original measurement operator).
%The main theorem is an extension of this
%argument to two-sided error.

\subsection{Grover iterator}
\label{subsection:grover}
As is usual in all analysis of amplitude amplification, the main operator to
study is the Grover iterator~\cite{Grover1996,BHMT}. Suppose we have a circuit $C$ acting
on an input state $\ket{\psi}$ and supposed to be measured using a two-output
projective measurement operator $\Proj=\langle P_E, I-P_E \rangle$. We consider a generalized version,
similar to the one studied by H{\o}yer~\cite{Hoyer2000}: $G(C,\ket{\psi}, \Proj, \theta,\alpha)
= C S_{\ket{\psi}} C^\dagger S_\Proj C$ using these additional gates: 
$S_{\ket{\psi}} = I - (1-e^{\I\theta})\ketbra{\psi}{\psi}$ and $S_\Proj = I -
(1-e^{\I\alpha})P_E$.

Let $\ket{\psi'}=C\ket{\psi}$ denote the output state, $\ket{\psi_E} = P_E
\ket{\psi'}$
%and $\ket{\psi_F} = (I-P_E)\ket{\psi'}$ denote the unnormalized
%post-measurement states after applying $C$ on input. 
and $p$ denote $\braket{\psi_E}{\psi_E}$ -- the probability of measuring outcome
$E$ for this output state.
%; since $P_E^\dagger P_E = P_E$, this also gives us $\braket{\psi'}{\psi_E}=p$.

%Let
%$\ket{\psi''}$ denote the output state of $\Q_1$ on $\ket{\psi}$.
It is easy to see that $C S_{\ket{\psi}} C^\dagger  
%= C (I - (1-e^{\I\theta})\ketbra{\psi}{\psi}) C^\dagger
= I - (1-e^{\I\theta})\ketbra{\psi'}{\psi'}$ and 
$S_\Proj C \ket{\psi} = \big(I - (1-e^{\I\alpha})P_E\big)\ket{\psi'}$.
One can then compute $\ket{\psi''}=G \ket{\psi}$ as 
%Thus we get the output state after applying Grover iterator as \\$\ket{\psi''}$ = $\big( I - (1-e^{\I\theta})\ketbra{\psi'}{\psi'} \big) \big( I -
%(1-e^{\I\alpha})P_E \big) \ket{\psi'}$ = $\big( e^{\I\theta} +
%(1-e^{\I\alpha})(1-e^{\I\theta})\braket{\psi'}{\psi_E} \big) \ket{\psi'} -
%(1-e^{\I\alpha})\ket{\psi_E}$ = 
$\big( e^{\I\theta} +
(1-e^{\I\alpha})(1-e^{\I\theta})p \big) \ket{\psi'} -
(1-e^{\I\alpha})\ket{\psi_E}$
and $P_E \ket{\psi''} = \big( e^{\I\theta} + e^{\I\alpha} - 1 +
(1-e^{\I\alpha})(1-e^{\I\theta})p \big) \ket{\psi_E}$.

We get the following lemma summarizing the relative increase in probability after
one application of our Grover iterator.
We will use $p'(\theta,\alpha,p)$ to denote the new probability of
measuring outcome $E$ on the output state
after applying $G$ on input $\ket{\psi}$.

\begin{lemma}\label{lemma:grover-iterator}
    Given a quantum system $\Q_1 = \langle \ket{\psi}, C, \Proj \rangle$ and
    $\alpha,\theta \in [0,\pi]$, let $G$ be
    the circuit for the Grover iterator $G(C,\ket{\psi}, \Proj, \theta,\alpha) =
    C S_{\ket{\psi}} C^\dagger S_\Proj C$. If $p$ denotes the probability of observing
outcome $E$ for $\Q_1$ and $p'$ denotes the same probability for the QS $\langle
\ket{\psi}, G, \Proj \rangle$, then $p' = p\Delta$ where $\Delta = \left|\big( e^{\I\theta} + e^{\I\alpha} - 1 +
(1-e^{\I\alpha})(1-e^{\I\theta})p \big)\right|^2$.
\end{lemma}

First, $p=0$ if and only if $p'=0$ which means amplification has no effect on
impossible outcomes.
On the other hand, if $p > 0$, $p'$ is maximized when $\theta=\alpha$;
it can be shown that
$\Delta = \big( (1-2p)\cos\theta - 2(1-p) \big)^2 +
\sin^2\theta$ in that case.
We will use $\Delta^*_p$ to denote the maximum value of $\Delta$ for any $p$ and
using optimal $\theta$ and $\alpha$. The corresponding {\em optimal Grover iterator} will be
denoted as $G^*_p(C,\ket{\psi}, \Proj)$; note that $G^*$ increases the probability from
$p$ to $p' = p \Delta^*_p$. Table~\ref{table:delta} summarizes the optimum value
of $p'$ and the relative increase for different possible values of initial
probability $p$.
Details of the relevant calculations are given in Appendix.

The following definition and corollary essentially describes the optimum
$\BB$-transform.

\begin{definition}[Optimal $\BB$-transform]\label{defn:opt-bb-transform}
$\BB_{p}^* : \big\langle \ket{\psi}, C, \Proj \big\rangle \longrightarrow
    \big\langle \ket{\psi}, G^*_p\big(C, \ket{\psi}, \Proj \big), \Proj
    \big\rangle$
\end{definition}

\begin{corollary}\label{cor:opt-bb-transform}
    If the output distribution of a QS $\Q$ is
    $\mu_\epsilon$, then the output distribution of $\BB^*_\epsilon(\Q)$ is $\langle \epsilon \Delta^*_\epsilon,
    1-\epsilon\Delta^*_\epsilon \rangle$. On the other hand, if the output distribution is
    $\mu_\delta$ (for some $\delta < \epsilon$), then the output distribution of $\BB^*_\epsilon(\Q)$
    is $\langle \delta',    1-\delta' \rangle$ for some $\delta' \ge \delta$ which
    can be computed using $\delta$ and $\epsilon$. Furthermore, $\delta =
    \delta'$ if and only if $\delta = 0$ (in which case, $\delta' = 0$).
\end{corollary}

\begin{table}
    \begin{tabular}{|c|c|c|c|}
	\hline
	\parbox[c][.7cm]{2.5cm}{\centering Range of initial probability $p$} & Optimum $\alpha=\theta$ & Relative increase
	$\frac{p'}{p}=\Delta^*_p$ & \parbox[c][1cm]{3cm}{\centering Amplified probability $p' = p \Delta^*_p$}\\
	\hline
	$p = 0.5$ & $\pi/2$ & 2 & 1 \\
	\hline
	$0.25 \le p \le 0.5$ & $\arccos \left( 1-\frac{1}{2p} \right)$ &
	$\frac{1}{p}$ & 1\\
	\hline
	$p \le 0.25$ & $\pi$ & $(3-4p)^2 \ge 4$ & $p(3-4p)^2 \ge 4p$ \\
	\hline
    \end{tabular}
    \caption{Optimum Grover iterator for different values of initial
probability\label{table:delta}}
\end{table}

In the next few subsections, we prove Separability
Lemma for different values of $\epsilon$.

\subsection{$\BB_\epsilon$ for $\epsilon \in [1/4,1/2]$}\label{subsection:bb-half}
This is the simplest of all cases, to $\BB$-transform
$(\delta,\epsilon)$-separable $\S_1$ to a $(\delta',1)$-separable one, for any
$1/4 \le \epsilon \le 1/2$ and for some
$\delta \le \delta'$. We can clearly use $\BB_{\epsilon} = \BB^*_{\epsilon}$ defined in
Definition~\ref{defn:opt-bb-transform}.
Separability Lemma immediately follows from Corollary~\ref{cor:opt-bb-transform} and
    Table~\ref{table:delta}.

%Let $\bigg\langle
%\ket{\psi}, C, \Proj \bigg\rangle$ be a quantum system from $\S_1$.
%The outcome probability being either $\delta$ or $\half$ and since in the latter
%case $G^*$ will increase the probability to $\half |\Delta^*_{\half}|^2=1$, we suggest this transformation.
%\begin{align}
%\BB_{\half} : \bigg\langle \ket{\psi}, C, \Proj \bigg\rangle & \longrightarrow
%    \bigg\langle \ket{\psi}, G^*\big(C, \ket{\psi}, \Proj \big), \Proj \bigg\rangle
%\label{eqn:bb-half}
%\end{align}

%\begin{proof}[Proof of Separability Lemma:]
%The $\BB$-transform given in Equation~\ref{eqn:bb-half} clearly converts a 
%QS with output distribution $\langle \half, \half \rangle$ to a QS with output
%distribution $\langle 0,1 \rangle$. Therefore, any
%$(\delta,\frac{1}{2})$-separable $\S_1$ is converted to a $(\delta',1)$-separable
%$\BB(\S_1)$ for some $\delta \le \delta' < 1$ where $\delta'=0$ if and only if
%$\delta=0$. Finally note that the transformation does not depend upon $\delta$.
%\end{proof}

\subsection{$\BB_{\epsilon}$ for $\epsilon > \half$}\label{subsection:bb-half+}
We use the idea proposed by Brassard et al.~\cite{BHMT} to first convert
$\S_1$ to a $(\delta',\half)$-separable $\S_2$; let $\BB^+_\epsilon$ denote
this transformation which is illustrated in Equation~\ref{eqn:bb-half+}.
This involves an
additional qubit in state $\ket{0}$ and an additional projective operator 
$\Proj_\epsilon = \langle P_\epsilon^0, I - P_\epsilon^0 \rangle$,
where, $$P_\epsilon^0 = \tfrac{1}{2\epsilon}\ketbra{0}{0} +
\sqrt{1-\tfrac{1}{2\epsilon}}\sqrt{\tfrac{1}{2\epsilon}}\ketbra{1}{0} + 
\sqrt{1-\tfrac{1}{2\epsilon}}\sqrt{\tfrac{1}{2\epsilon}}\ketbra{0}{1} + 
\left(1-\tfrac{1}{2\epsilon}\right) \ketbra{1}{1}$$

%The above projector was chosen so that the probability of measuring outcome $0$
%on the state $\ket{0}$ is $\frac{1}{2\epsilon}$ (since $\half < \epsilon \le 1$,
%$\half \le \frac{1}{2\epsilon} < 1$).
%This leads to the following transformation $\BB^+_\epsilon$.
%\begin{align}
%    \BB^+_\epsilon : \big\langle \ket{\psi}, C, \Proj \big\rangle &
%    \longrightarrow \big \langle \ket{\psi} \otimes \ket{0}, C \otimes I, \Proj
%    \otimes \Proj_\epsilon \big\rangle
%    \label{eqn:bb-half'}
%\end{align}

Then we convert $\S_2$ to a $(\delta'',1)$-separable $\S_3$ by using
$\BB_{\half}$ (see Subsection~\ref{subsection:bb-half}). Combining both of these, we
propose the following transformation for $\BB_\epsilon$. Here $\Proj'$ denotes
$\Proj \otimes \Proj_\epsilon$.
\begin{align}
    \big\langle \ket{\psi}, C, \Proj \big\rangle &
    \stackrel{\BB^+_\epsilon}{\longrightarrow} \big \langle \ket{\psi} \otimes \ket{0}, C \otimes I, \Proj' \big\rangle
    \stackrel{\BB_{\half}}{\longrightarrow} \big \langle \ket{\psi} \otimes \ket{0}, G^*_{1/2}\big(C \otimes I,
    \ket{\psi}\otimes \ket{0}, \Proj' \big), \Proj' \big\rangle
    \label{eqn:bb-half+}
\end{align}

\begin{proof}[Proof of Separability Lemma:]
    The transformation from $\S_2$ to $\S_3$ was shown to be correct in
    Subsection~\ref{subsection:bb-half}. Correctness of
    $\BB^+_\epsilon$ follows from the fact that the probability of measuring
    outcome $0$ on the state $\ket{0}$ is $\frac{1}{2\epsilon}$ (since $\half <
    \epsilon \le 1$, $\half \le \frac{1}{2\epsilon} < 1$). Let $p$ denote the probability of measuring outcome $E$
    for some $\Q = \langle \ket{\psi}, C, \Proj \big\rangle \in \S_1$ and let
    $p'$ denote the same
    probability for the QS $\langle \ket{\psi} \otimes \ket{0}, C \otimes I, \Proj
    \otimes \Proj_\epsilon \rangle$ of $\S_2$. Observe that, if $p=0$, then
    $p'=0$; furthermore, if $p = \epsilon > \half$, then $p'=\epsilon
    \frac{1}{2\epsilon} = \half$. Of course, the transformation does not
    depend upon $\delta$.
\end{proof}

\subsection{$\BB_{\epsilon}$ for $\epsilon < \quarter$}\label{subsection:bb-small}
To transform $(\delta,\epsilon)$-separable $\S_1$ to $(\delta',1)$-separable one,
we first repeatedly apply the optimum Grover iterator enough number of times to 
amplify $\epsilon$ beyond $\quarter$ and then apply a suitable
$\BB_{\epsilon_k}$ from Subsection~\ref{subsection:bb-half}.

Suppose $\epsilon < 1/4$.
Let $\epsilon_1 = \epsilon \Delta^*_\epsilon$, $\epsilon_2
= \epsilon_1 \Delta^*_{\epsilon_1}$, $\epsilon_3 = \epsilon_2
\Delta^*_{\epsilon_2}, \cdots$. Let $k$ be the smallest integer such that
$\epsilon_k \ge 1/4$; clearly, $\epsilon_1, \ldots, \epsilon_{k-1} < 1/4$ and
$\epsilon_k \in [1/4,1/2]$.
We define $\BB_\epsilon$ as the $k$ transformations $\BB_{\epsilon}^*,
\BB^*_{\epsilon_1}, \BB^*_{\epsilon_2}, \ldots \BB^*_{\epsilon_{k-1}}$ applied successively and then
followed by $\BB_{\epsilon_k}$.
\begin{align*}
    \BB_\epsilon: \langle \ket{\psi}, C, \Proj \rangle & \substack{~\BB^*_\epsilon~\\\longrightarrow} \langle
    \ket{\psi}, C_1, \Proj \rangle & \mbox{output
dist.}=\pair{\epsilon_1}{1-\epsilon_1} \mbox{ \& } C_1 = G^*_\epsilon(C,
    \ket{\psi},\Proj)\\
    & \substack{~\BB^*_{\epsilon_1}~\\\longrightarrow} \langle \ket{\psi}, C_2,
    \Proj \rangle & \mbox{output dist.}=\pair{\epsilon_2}{1-\epsilon_2}
    \mbox{ \& } C_2 = G^*_{\epsilon_1}(C_1,
    \ket{\psi},\Proj)\\
    & \substack{~\BB^*_{\epsilon_2}~\\\longrightarrow} \cdots & \dots \\
    & \substack{~\BB^*_{\epsilon_{k-1}}~\\\longrightarrow} \langle \ket{\psi}, C_k,
    \Proj \rangle & \mbox{output dist.}=\pair{\epsilon_k}{1-\epsilon_k}
    \mbox{ \& } C_{k} = G^*_{\epsilon_{k-1}}(C_{k-1},
    \ket{\psi},\Proj)\\
    & ~\substack{\BB_{\epsilon_k}\\\longrightarrow} \langle
    \ket{\psi'}, C_{k+1}, \Proj' \rangle &
\end{align*}

\begin{proof}[Proof of Separability Lemma:]
    Satisfiability Lemma is easily proved by observing that $\epsilon_k \in
    [1/4,1/2]$ and so, applying $\BB_{\epsilon_k}$ (from
    Subsection~\ref{subsection:bb-half}) at the last step ensures that the final QS
    has output distribution $\pair{1}{0}$.
It is also easy to check that these output distributions remain unchanged if and only
if $\delta = 0$.
\end{proof}

\subsection{Performance Evaluation}
Even though we propose a recursive approach to reduce error-probability of exact
error quantum systems, we show that our approach is essentially same as the
existing iterative approaches for amplitude amplification in terms of the number
of calls to $C$ and $C^\dagger$.

Take any quantum system $QS =
\langle \ket{\psi}, C, \Proj \rangle$. The existing approaches~\cite{BHMT,Hoyer2000} repeatedly apply the iterative Grover
operator
$\Q = (C S_{\ket{\psi}} C^\dagger S_\Proj)$ (generalized to act on input encoded as the initial state and output
state to be measured by any projective operator) on $C \ket{\psi}$. Here $S_{\ket{\psi}}$ and $S_\Proj$
modify the phase of certain states by $\theta=\alpha=\pi$ as specified in
Subsection~\ref{subsection:grover}.

Let $\epsilon$ denote the probability of
observing outcome $E$; let $\beta \in [0,\pi/2]$ be such that $\sin^2 \beta =
\epsilon$. Then, the probability of observing $E$ on repeated applications of
$\Q$, say $b$ times, on $C \ket{\psi}$  (i.e., on the output state of $\Q^b C
\ket{\psi}$) can be shown to be $\sin^2 \big( (2b+1)\beta \big)$.

As shown in Table~\ref{table:delta}, suitably choice of phases in
$S_{\ket{\psi}}$ and $S_\Proj$ can amplify any $\epsilon \in [0.25,1]$ to 1
using a $\BB$-transform that effectively corresponds to one application of
$\Q$ on $C\ket{\psi}$. So, if $\epsilon \ge 0.25$, our recursive method and the
iterative approach are identical.

So, we will now analyze $\BB_{\epsilon}$ for $\epsilon < 0.25$, in fact, for $\epsilon
\ll 0.25$. Let $k$ be the number of $\BB^*$-transforms required.
Recall from Subsection~\ref{subsection:bb-small} that $\BB_\epsilon$ keeps the input
and the projective operator unchanged and converts $C$ to some $C_{k+1}$ via
intermediate circuits $C_1, C_2, \ldots, C_k$ where
$C_{j+1} = G^*_{\epsilon_j}(C_j,
\ket{\psi}, \Proj)$ for $\epsilon < \epsilon_1 < \ldots < \epsilon_k \in
[1/4,1/2]$. The $S_{\ket{\psi}}$ and $S_\Proj$ operators in those $G^*$ are defined
using phases $\theta=\alpha=\pi$ as per Table~\ref{table:delta}.

\begin{restatable}{lemma}{ckinduction}\label{lemma:c_k_induction}
    For any $j \in [1,k]$, $C_j = \Q^{\frac{3^j-1}{2}} C$.
\end{restatable}

This lemma can be easily proved by induction on $k$ (see Appendix).
It shows that the final circuit obtained by our recursive approach is identical
to that obtained by apply a fixed $\Q$ a certain number of times.
Therefore, $\epsilon_{k} = \sin^2 \left( 3^k \beta
\right)$ which must be at least $1/4$. This stipulates that $k \ge \log_3
\frac{\pi}{6\beta}$. The total number of calls to $C$ and $C^\dagger$ made by
our recursive algorithm to amplify $\epsilon < 0.25$ to some $\epsilon_k > 0.25$
can then be easily shown to be $1+\frac{\pi}{3\beta}$ (rather, the next higher
integer) -- which is exactly the same as that in $\Q^{(3^k-1)/2}C$.

\subsection{Proof of Theorem~\ref{theorem:main-separable}}
We are now ready to prove Theorem~\ref{theorem:main-separable} using
Separability Lemma.
We will use the following notation. If $\BB$ is a transformation for a set of
quantum systems $\S$, then the set of {\em transformed quantum systems} after
applying $\BB$ will be denoted by $\BB(\S)$.

%{\bf Theorem:}    Let $\S = \{\Q_1, \Q_2, \ldots \}$ be a collection of quantum systems such
%    that any of its quantum systems have output distribution either $\mu_\delta$
%    or $\mu_\epsilon$ for some $\delta < \epsilon$.
%   
%    Then $\S$ is perfectly-solvable via a $\BB$-transition
%    $\BB_{\delta,\epsilon}$, i.e., any $\Q_i \in \S$ can be transformed to some
%    $\Q'_i$ such that:
%    \begin{itemize}
%	\item if output distribution of $\Q_i$ is $\mu_\delta$, then outcome of
%	    $\Q'_i$ is never $E$ and
%	\item if output distribution of $\Q_i$ is $\mu_\epsilon$, then outcome of
%	    $\Q'_i$ is always $E$.
%    \end{itemize}
\begin{proof}
    The given $\S$ in the theorem is $(\delta,\epsilon)$-separable. Our required
    $\BB_{\delta,\epsilon}$ will be composed of a series of $\BB$-transforms:
    $\BB_\epsilon$, $\BB_2$ and $\BB_\delta$.

    $\BB_\epsilon$ is chosen such so as to solve $QD(\S)$ with error $(\delta',1)$ for
    some $\delta < \delta'$. This step can skipped ($\BB_\epsilon$ can be set to
    identity) if $\epsilon = 1$; on the other hand, if $\epsilon < 1$,
    we can use $\BB_\epsilon$ from Lemma~\ref{lemma:fully-separable}, which
    implies that $\BB_\epsilon(\S)$ is
    $(\delta',1)$-separable for some $\delta'$ (that depends on $\delta$ and
    $\epsilon$). Let $\S_1$ denote $\BB_\epsilon(\S)$.

    $\BB_2$ is the following transform: $\llangle \ket{\psi}, C, (P_1, P_2)
    \rrangle \longrightarrow \llangle \ket{\psi}, C, (P_2, P_1) \rrangle$.
    Let $\S_2 = \BB_2(\S_1)$. Any $QS \in \S_1$ with $\mu(QS)=\mu_{\delta'}$ is
    transformed to $QS' \in \S_2$ with $\mu(QS') = 1-\delta'$ and similarly, if
    $\mu(QS) = \mu_1$, then $\mu(QS') = \mu_0$. Therefore, $\S_2$ is $(0,1-\delta')$-separable.

    By property of $\BB_\epsilon$, $\delta=\delta'=0$ if and only if $\delta=0$ and in
    that case, we have obtained $(0,1)$-separable $\S_2$. On the other hand,
    if $\delta > 0$, then $\delta' > 0$. Let $\delta''$ denote
    $1-\delta'$.
    Since $\S_2$ is $(0,\delta'')$-separable, apply
    Lemma~\ref{lemma:fully-separable} again to get $\BB_{\delta}$ such that
    $\S' = \BB_{\delta}(\S_2)$ is $(0,1)$-separable.

    Our required transform $\BB$ is a sequential application of $\BB_\epsilon$ followed
    by $\BB_2$ followed by $\BB_{\delta}$. As explained above,
    $\BB_\delta(\BB_2(\BB_\epsilon(\cdot)))$ is a $(0,1)$-solution of $QD(\S)$.
\end{proof}

\section{Uniform version of
Theorem~\ref{theorem:main-separable}}\label{section:uniform}
The non-uniformity in Definition~\ref{defn:bb-transform} is not very helpful if
we wish to obtain a true black-box extension of a quantum system $\Q = \llangle
\ket{\psi},C,\Proj \rrangle$.
Note that the extension to the input qubits and the
extension to the projective measurement operator is anyway independent of $\Q$
and $n$, the gates in $C'$ are uniform in $n$,
and furthermore, the transformed circuit $C'$ is allowed to call the original
circuit $C$ (and its inverse $C^\dagger$) in a black-box manner; however, some
of the gates in $C'$ may additionally depend upon $\ket{\psi}$ and operators of
$\Proj$. It would be really good to obtain a more uniform conversion which
necessitates the following definition.

\begin{definition}[Uniform $\BB$-transform]
    A $\BB$-transform for converting multiple QS $\{\Q_1, \Q_2, \ldots \}$ is said
    to be {\em uniform} if the circuit of $\BB(\Q_i)$ is identical for all source $\Q_i$ except for the
    calls to $C$ and $C^\dagger$ corresponding to $\Q_i$.
\end{definition}

\subsection{Uniform Grover iterator}
We want to study some
sufficient conditions for the $\BB$-transforms to be uniform by constructing a
uniform version of Grover iterator.

Since Grover iterator uses $\S_\Proj$, it is crucial to have identical measurement operators for
all quantum systems. This is, however, not such a major requirement since it is always
possible to change measurement operators by extending a quantum circuit with
suitable operators.

Except the gates $S_{\ket{\psi}} = I - (1-e^{\I\theta})\ketbra{\psi}{\psi}$ which depend upon the
corresponding input to the circuit ($\ket{\psi}$), none of the other gates
used in $\BB$-transforms that are involved in the
proof of Theorem~\ref{theorem:main-separable} depend upon the input
state (see Section~\ref{section:proof}).
However, a $\BB$-transform may still become uniform if all the inputs in $\S_1$,
and hence all such $S_{\ket{\psi}}$ gates, will be identical.

Now consider a second option -- all measurement operators are identical and all the input states are not identical but
they form an orthonormal set. We show that it is still possible to
apply $S_{\ket{\psi}}$ in a uniform manner. Recall that this gate changes the
phase of any state depending upon whether it is $\ket{\psi}$ or not and the main
difficulty appears to be the fact that the input state cannot be copied and 
stored for a later application of the conditional phase gate. So our main
idea is to convert $\ket{\psi}$ to some state in the standard basis since 
it is possible to copy and store states in the standard basis using the {\em quantum
fanout gate}~\cite{Durr1999}. This gate copies a standard basis state to
another register: $F_m \ket{x_1 \ldots x_m}\ket{b_1 \ldots b_m} = \ket{x_1 \ldots
x_m}\ket{(x_1 \oplus b_1) \ldots (x_m \oplus b_m)}$ for $x_1 \ldots x_m \in
\{0,1\}^m$ and $b_1 \ldots b_m \in \{0,1\}^m$ shows the operation for
``copying'' $m$-qubits.

\begin{figure}[!htb]
\begin{subfigure}[b]{1.2in}
    {\raisebox{8mm}{\resizebox{1.2in}{!}{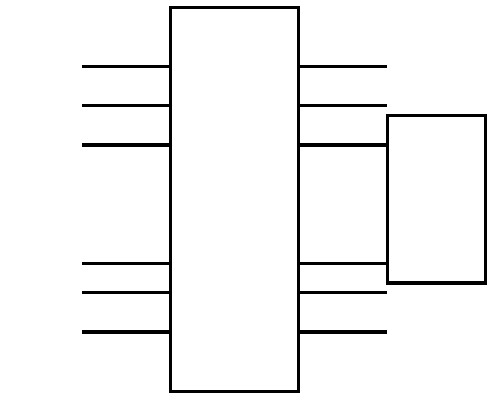}}}%
    \caption{$S_{\ket{\psi}}$\label{fig:s-psi-uniform-a}}
\end{subfigure}%
~
\begin{subfigure}[b]{3in}
    %\subfloat[Uniform $S_{\ket{\psi}}$]
    {\resizebox{3in}{!}{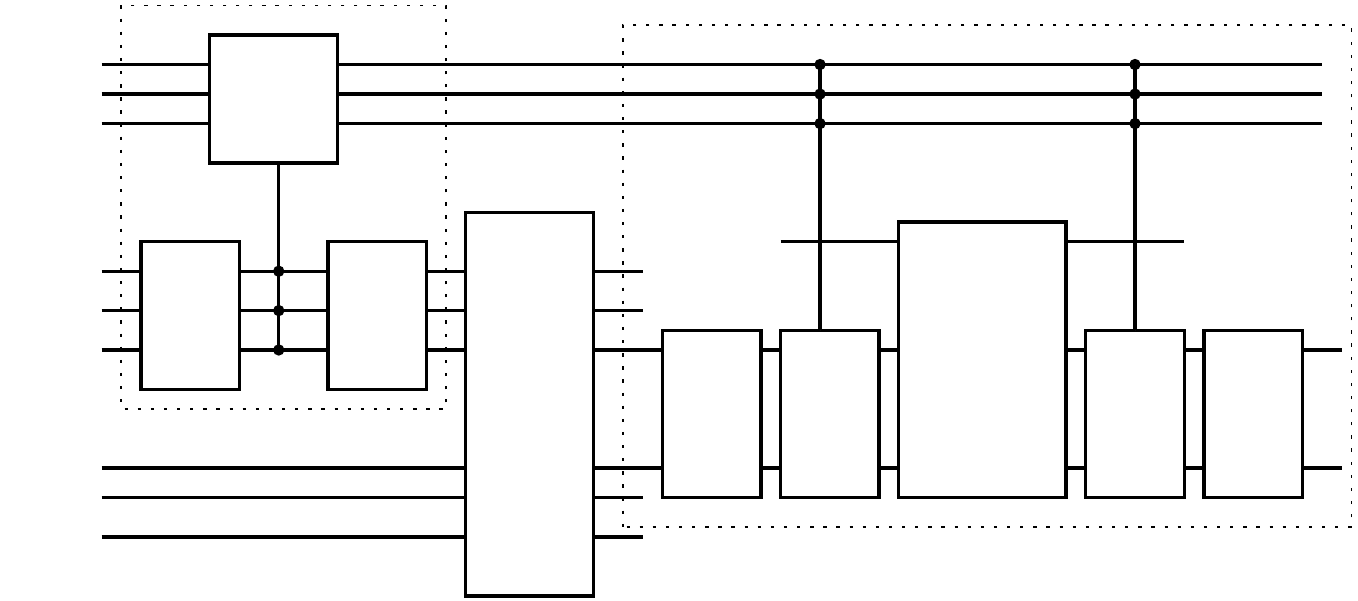}}%
    \caption{Uniform $S_{\ket{\psi}}$\label{fig:s-psi-uniform-b}} %
\end{subfigure}%
~
\begin{subfigure}[b]{1.2in}
    %\subfloat[$S_\theta$ operator]
    {\raisebox{1cm}{\resizebox{1.2in}{!}{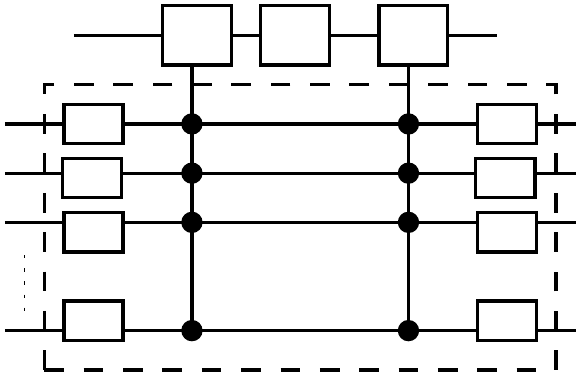}}}%
    \caption{$S_\theta$ operator \label{fig:s-psi-uniform-c}}
\end{subfigure}
    \caption{\label{fig:s-psi-uniform} Applying operator $S_{\ket{\psi}}$ in a
	uniform manner. Figure~\ref{fig:s-psi-uniform-a} shows the non-uniform
	operator and Figure~\ref{fig:s-psi-uniform-b} shows its uniform version
    (dotted box on the left shows initialization and dotted box on the right
    shows $S_{\ket{\psi}}$ being applied uniformly).
Figure~\ref{fig:s-psi-uniform-c} shows the $S_\theta$ operator from
Figure~\ref{fig:s-psi-uniform-b}.}
\end{figure}

See Figure~\ref{fig:s-psi-uniform} for a uniform version of $\S_{\ket{\psi}}$.
Figure~\ref{fig:s-psi-uniform-a} shows $S_{\ket{\psi}}$ as a part of an arbitrary quantum
circuit, say $C$ that takes as input an $m$-qubit state $\ket{\psi}$ (and some
ancill\ae) and $\S_{\ket{\psi}}$, on $m$-qubits, is one of its gates
Since we are now considering the case that $C$ is applied only on orthogonal
input states (suppose denoted by $\ket{\psi_1}, \ket{\psi_2},
\cdots$), therefore, there exists a one to one mapping between these states
and a subset of the $m$-qubit standard basis states $\ket{1}, \ket{2}, \cdots$.
Let $U$ denote the unitary operator for the mapping, i.e., $U\ket{\psi_v} =
\ket{v}$.

Figure~\ref{fig:s-psi-uniform-b} illustrates a circuit $C'$ that applies $S_{\ket{\psi}}$ without
requiring a gate that explicitly depends upon $\ket{\psi}$. Apart from the two
registers of $C$ (the input $\ket{\psi}$ and ancill\ae\ qubits), $C'$ also uses
$m$ additional ancill\ae\ qubits in state $\ket{0}$. Other than the standard gates
($T$ stands for the unbounded fanout Toffoli and $X$ is the quantum NOT gate),
$C'$ uses three additional gates: $F_m$, $P_\theta$ and $S_\theta$. The $F_m$ gate is
the quantum fanout gate. $P_\theta$ changes phase
of $\ket{1}$ by $e^{\I\theta}$: $P_\theta = I - (1-e^{\I\theta})\ketbra{0}{0}$.
The $S_\theta$ gate uses
an additional reusable ancill\ae\ $\ket{0}$ and changes the phase by
$e^{\I\theta}$ only for the state $\ket{0^m}$ (illustrated in
Figure~\ref{fig:s-psi-uniform-c}).

The state of the first two registers after the left dotted box in
Figure~\ref{fig:s-psi-uniform-b} is simply $\ket{0^m}\ket{\psi} \to
\ket{v}\ket{\psi}$ where $\ket{v}$ is the standard basis vector $U\ket{\psi}$.
We will next analyze the operator for the right dotted
box, say denoted by $U_R$. $S_\theta$ can be written as $I -
(1-e^{\I\theta})\ketbra{0^m}{0^m}$ and the $F_m$ operator essentially behaves
like $F_m \ket{b_1 \ldots b_m} \to \ket{(v_1 \oplus b_1), \ldots (v_m
\oplus b_m)}$.
The following calculation (for only the qubits involved) shows that the operator for the right dotted box is
identical with $S_{\ket{\psi}}$. 
\begin{align*}
U_R = & (I \otimes U^\dagger) F_m (I \otimes S_\theta) F_m	(I \otimes U)
= (I \otimes U^\dagger) F_m \big(I \otimes (I -
(1-e^{\I\theta})\ketbra{0^m}{0^m}) \big) F_m (I \otimes U)\\
= & (I \otimes U^\dagger) \big(I \otimes (I -
(1-e^{\I\theta})\ketbra{v}{v}) \big) (I \otimes U)
= I \otimes (I - (1-e^{\I\theta})\ketbra{\psi}{\psi}) = I \otimes
S_{\ket{\psi}}
\end{align*}

The results of this subsection can be summarized in the following lemma.

\begin{lemma}\label{lemma:uniform-bb-transform}
    The $\BB$-transform in Theorem~\ref{theorem:main-separable} can be made uniform if
    all projection operators in the quantum systems of $\S$ are identical and
    all input states in $\S$ are either identical or form an orthonormal set of
    states.
\end{lemma}

\section{Distinguishing two circuits}\label{section:distinguish-circuits}
Suppose we are given a quantum circuit $C$ (as black-box) and two
different operators $C_1$ and
$C_2$, all acting on the same Hilbert space, and we are told that the operator for $C$ is either $C_1$ or $C_2$. We have to
determine $C$ corresponds to which one. We assume that we also have access to its inverse operator $C^\dagger$.

The analogous problem for deterministic (classical) functions is
trivial. Two distinct functions must differ at some input which can be
determined from their function descriptions (the problem is NP-hard
but we are not concerned about feasibility, not efficiency, for this discussion). The output of
$C$ on this input will identify whether $C$ is $C_1$ or $C_2$.
However, if $C$ is a randomized circuit or
algorithm, then except for a few trivial cases, the output of $C$
generates a sample distribution over the output of $C_1$ and $C_2$; the
question of determining the correct distribution of $\C$ without any error is believed
to be hard, if not impossible.

However, it is possible to give a positive answer to the same question for
quantum circuits.
Select a suitable $\ket{\phi}$ and compute the two possible output states
${\ket{\psi_1} = C_1\ket{\phi}}$, ${\ket{\psi_2} = C_2\ket{\phi}}$. Choose
projective operators $\Proj=\langle I - \ketbra{\psi_1}{\psi_1},
\ketbra{\psi_1}{\psi_1} \rangle$ with respective outcomes $E$ and $F$.

Consider these two quantum systems: $\langle \ket{\phi},
C_1, \Proj \rrangle$ and $\langle \ket{\phi}, C_2, \Proj \rrangle$.
The output distribution of the first QS is $\pair{0}{1}$ and that of the second
is $\pair{\epsilon}{1-\epsilon}$ where $\epsilon =
1-|\braket{\psi_1}{\psi_2}|^2 > 0$.
%For the
%first QS, the probability of
%observing outcome $E$, when $C_1\ket{\phi}$ is measured using $\Proj$, is 0;
%so the output probability distribution is $\langle 0,1 \rangle$. On the other
%hand, for the latter QS the same probability is $1-|\braket{\psi_1}{\psi_2}
%|^2$. Denoting this probability as $\epsilon >0$, the output distribution in
%this case is $\langle \epsilon, 1-\epsilon \rangle$. Therefore, $\S_1$ is 
%$\epsilon$-separable.

Now, Theorem \ref{theorem:main-separable} can be applied on the QS $\langle
\ket{\phi}, C, \Proj \rangle$ which essentially gives us a circuit $C'$ (that
calls $C$ and $C^\dagger$) along with suitably extended input and
measurement operators,
with the property that if the outcome of the QS is $E$, then $C$ is surely $C_1$ and
otherwise $C_2$.

It is perfectly okay to use any $\ket{\phi}$ as the input state; however, since
the size of $C'$ depends inversely upon $\epsilon$ so it makes sense to have the
largest possible $\epsilon$. A recent result~\cite{2015BeraATPG} can be used to
determine the optimum initial state (details of this is presented in the
Appendix).

\paragraph*{Single-fault detection}
Fault detection is a major step in the workflow of circuit
fabrication. It is common in research and industry to assume that practically
most faults appear according to a few known fault models.
A standard approach to detecting if a circuit is faulty is to generate a set of
test patterns (inputs) such that the output of a fault-free circuit would be
different from that of a faulty-circuit. This method is known as ATPG (automatic
test-pattern generation) and is well-studied for classical
circuits and very
recently, seeing use even for quantum circuits~\cite{paler2011tomographic}.

ATPG is computationally difficult being NP-hard~\cite{IbarraSahniATPG}, and even harder
for quantum circuits because the measurement output of these circuits is
probabilistic, and hence even a single test pattern will generate a distribution
over possible outcomes.

However, the technique described earlier in this section can come to our rescue
in the special case of only one fault model, i.e., given a circuit $C$ as a
black-box unit, we wish to determine if $C$ is fault-free (i.e., $C=C_1$) or $C$
is faulty (with fault model $C_2$). We can reliably answer this question without
any chance of error using the approach described above.

\section{Exact Error Algorithms}\label{section:exact-error-classes}
Usual probabilistic classes like \RP and \BPP are defined in terms of errors that are upper bounded
by constants. They are rarely defined in terms of exact error, primarily due to
the lack of robustness in definition that accompanies this concept. There is no known
technique to show that the class of problems with one-sided error exactly same as $0.3$
remains unchanged if the error is instead $0.301$.
%Nevertheless, it is a reasonable question to examine these classes where errors
%need to be exactly some given constant.
%Conceptually, there
%should not be much difference between complexity of problems that admit randomized
%algorithms with one-sided error exactly, say, $0.31$ to that with error exactly
%$0.3$. However, one intuitive technique of decreasing error probabilities of
%classical probabilistic classes is via repeated trials which does not work as
%expected for exact error classes.
%
Consider, for example, the simplified class \ERP ($\P \subseteq \ERP \subseteq
\RP$) whose problems have randomized algorithms similar to those
for \RP, but with an additional requirement that the error is same for
all ``no'' instances (of any length). We similarly define \EBPP as the class of
problems with exact two-sided error polymomial-time algorithms.
Based on what we know, $\P \not= \ERP \not= \EBPP$.
However, we were able to prove that the
quantum analogs of these classes have identical complexity
using our generalization of quantum amplitude amplification.

\begin{definition}
    $\EBQP_{\delta,\epsilon}$ is the class of languages $L$ for which there
    exists a uniform family of polynomial-size quantum circuits $\{C_n\}$, a uniform family of
    states for $a_n$ ancill\ae\ qubits
$\ket{A_n}$ and a uniform family of two-outcome projective measurement operators
$\{\Proj_n\}$ such that $C_n$ and $\Proj_n$ act on a space of $n+a_n$ qubits and
the following hold for any $x \in \{0,1\}^n$, $\forall n$:
\begin{itemize}
    \item if $x \not\in L$, then the output distribution of $\llangle \ket{x}
	\otimes \ket{A_n}, C_n, \Proj_n \rrangle$ is $\mu_\delta$ (i.e., when
	the output state of $C_{n}$ on input state $\ket{x} \otimes \ket{A_n}$ is measured
	using $\Proj_n$, outcome $E$ is observed with probability $\delta$) and
    \item if $x \in L$, then the output distribution of $\llangle \ket{x}
	\otimes \ket{A_n}, C_n, \Proj_n \rrangle$ is $\mu_\epsilon$ (i.e.,
	outcome $E$ is observed with probability $\epsilon$ upon similar
	measurement as the above case).
\end{itemize}
\noindent$\ERQP_\epsilon$ is simply $\EBQP_{0,\epsilon}$.
Define $\displaystyle\EBQP = \bigcup_{\epsilon>\delta\ge 0}
    \EBQP_{\delta,\epsilon}$ and $\displaystyle\ERQP =
    \bigcup_{\epsilon>0} \ERQP_\epsilon$.
\end{definition}

Note that, unlike the usual definitions of probabilistic classes, for these
classes it is not even clear if the different classes
$\EBQP_{\delta,\epsilon}$ for different $\delta$ and $\epsilon$ are identical.
However, the following lemma is obvious from these definitions.
\begin{lemma}
    $\EQP = \EBQP_{0,1} = \ERQP_1$ and
    $\EQP \subseteq \ERQP \subseteq \EBQP$.
\end{lemma}

The main result of this section is a simple application of
Theorem~\ref{theorem:main-separable} and Lemma~\ref{lemma:uniform-bb-transform}.
\begin{theorem}\label{theorem:eqp-ebqp}
    $\EQP = \ERQP = \EBQP$.
\end{theorem}
\begin{proof}
We essentially need to show that $\EBQP \subseteq \EQP$. To prove this we will 
show that for any $L$, if $L \in
\EBQP_{\delta,\epsilon}$ (for any $\epsilon > \delta \ge 0$), then $L \in
\EBQP_{0,1}$.

Fix an arbitrary $n$. 
For any binary string $x$ of length $n$, define the quantum system $\Q_x = \llangle \ket{x} \otimes \ket{A_n}, C_n, \Proj_n \rrangle$
where $\ket{A_n}$, $C_n$ and $\Proj_n$ are obtained from the definition of
$\EBQP_{\delta,\epsilon}$ and the fact that $L \in \EBQP_{\delta,\epsilon}$.
Now consider these sets of quantum systems $\S_n = \{\Q_x ~:~ x \in \{0,1\}^n \}$ for all $n
> 0$. Clearly, there are two possible output distributions of any $\S_n$, namely,
$\mu_\delta$ and $\mu_\epsilon$. Since the input states in $\S_n$ are
orthonormal and the projection operators therein are identical, we can therefore apply
Theorem~\ref{theorem:main-separable} and Lemma~\ref{lemma:uniform-bb-transform}
to obtain a uniform transformation $\BB_{\delta,\epsilon}$ which
perfectly solves the problem of $QD(\S_n)$. Let $\BB_{\delta,\epsilon}(\Q_x) =
\Q'_x = \llangle \ket{x} \otimes\ket{A_n} \otimes \ket{00\ldots 0}, C'_n, \Proj'_n
\rrangle$ which gives us (i) a circuit $C'_n$ which calls $C_n$ (and
$C_n^\dagger$) (ii) a two-outcome projective measurement operator $\Proj'_n$ and
a (iii) set of ancill\ae\ qubits in state $\ket{\mathbf{00\ldots 0}}$
such that the following holds for the outcome of $C'_n$ on $\ket{x} \otimes
\ket{A_n} \otimes \ket{00\ldots 0}$ when measured using $\Proj'_n$.
\begin{itemize}
    \item If $x \not\in L$, then the output distribution of $\Q'_x$ is $\mu_0$,
	i.e., the outcome is never $E$.
    \item If $x \in L$, then the output distribution of $\Q'_x$ is $\mu_1$,
	i.e., the outcome is always $E$.
\end{itemize}

Therefore, we get a uniform family of circuits $\{C'_n\}$, a uniform family of
ancill\ae\ qubits $\ket{A_n} \otimes \ket{\mathbf{00\ldots 0}}$ and a uniform family of two-outcome
projective measurement operator $\{\Proj'_n\}$ such that the outcome of
$C'_{|x|}$ on any $\ket{x}$, with additional ancill\ae\ qubits in a uniformly
generated state, when measured by $\Proj'_{|x|}$ indicates whether $x \in L$
without any probability of error. 
Since $C_n'$ uses constantly many calls to $C_n$ and $C_n^\dagger$ along with other
gates (the constant depends only on 
$\delta$ and $\epsilon$), this shows that $L \in \EBQP_{0,1}$.
\end{proof}

\begin{figure*}[!h]
\centering \resizebox{0.65\linewidth}{!}{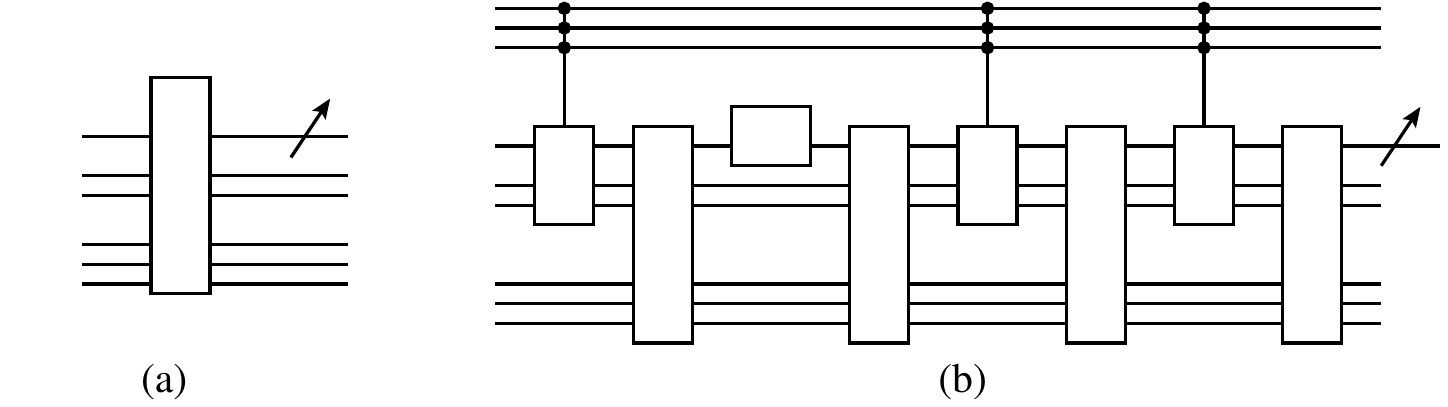} \qquad %
\resizebox{0.25\linewidth}{!}{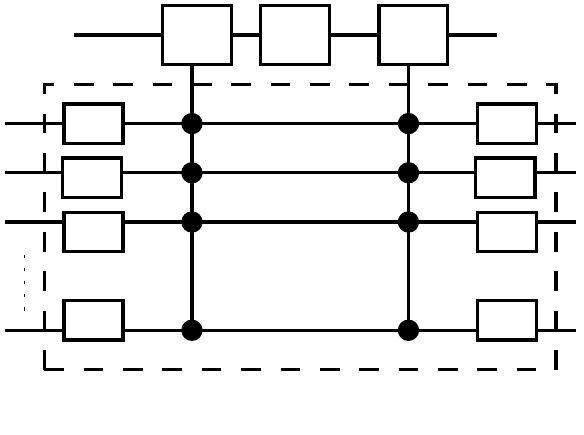}
\caption{Circuit for $\C'$ (left) and $S_0$ gate in $\C'$ (right)}
\label{fig:erqp-half}
\end{figure*}

We illustrate an application of the above theorem to obtain an error-free
circuit for an $\ERQP_{1/2}$ language $L$ (see Appendix for an explicit proof). Consider circuit $C$ in
Figure~\ref{fig:erqp-half}(a) which can identify if $x \in L$ with one-sided
error $0.5$. As is typical in quantum circuits, in this example only one of the output qubits of
the circuit is measured in the standard basis ($P_E = \ketbra{0}{0} \otimes I$
and $\Proj = \langle P_E, 1-P_E \rangle$); therefore, if $x \not\in L$, then
the output qubit is never observed in state $\ket{0}$ and if $x \in L$, then the
output qubit is observed in states $\ket{0}$ or $\ket{1}$ with equal
probability. The circuit $C'$ shown in Figure~\ref{fig:erqp-half}(b) shows how to
remove the probability of error; the same output qubit is measured in the
standard basis for outcome and some additional qubits in state $\ket{0}$ are
used as ancill\ae. Apart from calling $C$ and $C^\dagger$, $C'$ uses the
n-qubit Fanout gate $F_n$, a
conditional phase gate $S_0$ 
\footnote{$S_0
\ket{0 0 \ldots 0} = \I \ket{00\ldots 0}$ and for other states $S_0
\ket{x_1 x_2 \ldots x_k} = \ket{x_1 \ldots x_k}$ (illustrated in
Figure~\ref{fig:erqp-half}(c)).}
which changes phase of $\ket{00\ldots 0}$ by $\I$, and $P$
does the same to $\ket{1}$.

\subsection{Exact amplitude amplification
(Theorem~\ref{theorem:aa-two-sided})}
%We end this subsection by giving a quick sketch of the proof of
%Theorem~\ref{theorem:aa-two-sided}.

\begin{proof}[Proof of Theorem~\ref{theorem:aa-two-sided}]
Let $\Proj$
denote the two-outcome projective measurement operator used in the original
two-sided exact error circuit $C$. $C$ can be of two types depending on how
it accesses its input. Any input $x \in X$ can be accessed 
either through the input state $\ket{x}$ (along with ancill\ae\ initialized to $\ket{00\ldots}$, wlog.) or through an oracle gate $U_x:
\ket{x,b} \to \ket{x, b\oplus \Phi(x)}$ (for $b \in \{0,1\}$). If $C$ is of the
former type, then Theorem~\ref{theorem:aa-two-sided} is essentially same as
Theorem~\ref{theorem:eqp-ebqp}.

Next we focus on 
%So, we focus on the proof when $C$ is of the latter type -- 
circuits with oracle
gates. Let $C^{U_x}$ denote this circuit
when given
$U_x$ as the oracle gate corresponding to an input $x \in X$. The input state to
$C^{U_x}$ can be taken to be $\ket{00\ldots 0}$, wlog.
The proof follows by applying Theorem~\ref{theorem:main-separable} on this 
collection of quantum systems: $\left\{ \llangle
\ket{00\ldots}, C^{U_x}, \Proj \rrangle ~:~ x \in X \right\}$.

Observe that this collection satisfies the conditions of
Lemma~\ref{lemma:uniform-bb-transform}. So, the corresponding $\BB$-transform
is uniform which implies that all the transformed circuits
for these quantum systems are identical, except for the calls to
$C$ and $C^\dagger$. Therefore, we can choose this transformed oracle circuit as
our required $C'$ of Theorem \ref{theorem:aa-two-sided}.
\end{proof}

\section{Conclusion}
Is there a classical method that can accurately decide the distribution of a random
variable $X$ among two given distributions based on multiple samples of $X$?
Probably no. %, if we are looking for a completely error-free procedure. 
On the
other hand, if the random variables come from a quantum source, we show that
quantum circuits exist that can do the same without any probability of error.
A quantum circuit, along with an input state and a measurement operator, can be
consider as a quantum source of samples drawn over the distribution of
the measurement outcomes.

The underlying technique is a generalization of quantum amplitude
amplification to two-sided error and for circuits without oracle gates. We used
our amplification technique to distinguish between two circuits, when used as a
black box, which has application in fault detection of quantum circuits. We also
defined a restricted version of quantum one-sided and two-sided bounded error
classes and used generalized amplification to show that those complexity
classes collapse to (error-free) quantum polynomial time complexity
class.

It would be interesting to investigate if this approach can
be used for ATPG with more than one fault models and for amplifying
standard bounded-error classes $\BQP$ and $\RQP$.
%\bibliography{DJ}

\appendix
\section{Proof of $\ERQP_{1/2} \subseteq \EQP$}

\begin{lemma}\label{lemma:rqp_0.5_in_EQP}
    If a language $L \in \ERQP_{1/2}$, then $L \in \EQP$.
\end{lemma}

\begin{proof}%{Proof:}
We will assume that the algorithms end with
a measurement of a specified qubit in the computational basis -- this is
equivalent most other ways measurement strategies that are commonly applied.

Take any $L \in \ERQP_{1/2}$, and consider the corresponding circuit $C$
(illustrated in Figure~\ref{fig:erqp-half}(a)). Suppose $m$
denotes the number of ancill\ae\ qubits used by $C$, and $n$ denotes the length
of any input $x$, then $C$ acts on $\H^{\tensor n} \tensor \H^{\tensor m}$ and
its output is given by $\ket{\psi} = C
\ket{x}\ket{0^m}$. Without loss of generality, suppose that the first qubit is
specified for measurement, then the projective measurement operator applied is
$\ket{0}\bra{0} \tensor I$.

We will now a construct an \EQP\ circuit $\C'$ to decide the same language $L$. But
first note that, $\ket{\psi} = \ket{0}\ket{\psi_0} + \ket{1}\ket{\psi_1}$ and
that, if $x \not\in L$, $\braket{\psi_1}{\psi_1}=0$, and if $x\in L$,
$\braket{\psi_1}{\psi_1}=1/2$ ($=\braket{\psi_0}{\psi_0}$). The circuit is
constructed as 
$\C' = \A S_0 \A^{-1} P \A$ and described in Figure \ref{fig:erqp-half}(b).
$\C'$ acts on $\H^{\tensor n} \tensor \H^{\tensor n} \tensor \H^{\tensor m}$,
and we will denote the space as 3 registers $P,Q,R$, respectively, of $n,n,m$
qubits. The
gates will be labelled with the registers (as superscripts) they are applied on
in the following description.

Besides the circuit $C$, which will be used always on registers $QR$,
we will make frequent use of the {\em fanout} operator\cite{Durr1999}.
This, and the other components of $\C'$, are listed below.
\begin{itemize}
    \item The fanout operator effectively copies basis states from a control
	qubit to a target qubit. On two registers of $n$ qubits each, it works
	as $F_n\ket{a_1 \ldots a_n}\ket{b_1 \ldots b_n} =
	\ket{a_1 \ldots a_n}\ket{(b_1 \xor a_1) \ldots (b_n \xor a_n)}$.
	Note that, $F^\dagger_n = F_n$.
    \item $\A = (F_n^{PQ} \tensor I) \tensor (I \tensor C^{QR})$ 
    \item $P^{Q} = I - (1-\I)\ketbra{0}{0}$ is the phase gate $P$ applied on 
	the first qubit of register $Q$.
	Notice that, the first qubit of register $Q$ is the measurement qubit
	with respect to $C$.
    \item $S_0^{QR}=I-(1-\I)\ketbra{0^{n+m}}{0^{n+m}}$ which changes the phase
	of the basis state in which all qubits are in the
	state $\ket{0}$. Implementation of $S_0$ is shown in Figure
	\ref{fig:erqp-half}(c) -- it requires one additional qubit initialized to
	$\ket{0}$. However this qubit is in state $\ket{0}$ after application of
	this operator, so this qubit could be reused if required. This extra
	qubit has been left out in the description of $\C'$.
    \item The input to $\C'$ will be $\ket{x}\ket{0^{\tensor n}}\ket{0^{\tensor
	m}}$.
    \item We will measure the first qubit of register $Q$ in the standard basis
	at the end.
\end{itemize}

Next, we will describe the operation of $\C'$.
\begin{align*}
\C' \ket{x}\ket{0^n}\ket{0^m} 
= & C^{QR} \cdot F_n^{PQ} \cdot S_0^{QR} \cdot F_n^{PQ} \cdot C^{\dagger QR}
\cdot P^{Q} \cdot C^{QR} \cdot F_n^{PQ}
~~ \ket{x}\ket{0^n}\ket{0^m}\\
= & C^{QR} \cdot F_n^{PQ} \cdot S_0^{QR} \cdot F_n^{PQ} \cdot C^{\dagger QR}
\cdot P^{Q} \cdot C^{QR}
~~ \ket{x}\ket{x}\ket{0^n}\\
= & C^{QR} \cdot F_n^{PQ} \cdot S_0^{QR} \cdot F_n^{PQ} \cdot C^{\dagger QR}
\cdot P^{Q}
~~ \ket{x} \bigg(\ket{0}\ket{\psi_0} + \ket{1}\ket{\psi_1}\bigg)\\
= & C^{QR} \cdot F_n^{PQ} \cdot S_0^{QR} \cdot F_n^{PQ} \cdot C^{\dagger QR}
~~ \ket{x} \bigg(\ket{0}\ket{\psi_0} + \I\ket{1}\ket{\psi_1}\bigg)~~~~~(*)
\end{align*}

We will now simplify the remaining operator.
\begin{align*}
& C^{QR} \cdot F_n^{PQ} \cdot S_0^{QR} \cdot F_n^{PQ} \cdot C^{\dagger QR} \\
= & C^{QR} \cdot F_n^{PQ} \cdot \bigg(I-(1-\I)I^P\tensor\ketbra{0^{n+m}}{0^{n+m}}
\bigg) \cdot F_n^{PQ} \cdot C^{\dagger QR}\\
= & C^{QR} \cdot F_n^{PQ} \cdot \bigg(I-(1-\I)\sum_{\text{$n$-bit
$p$}}\ketbra{p,0^{n+m}}{p,0^{n+m}} \bigg) \cdot F_n^{PQ} \cdot C^{\dagger QR}\\
= & C^{QR} \cdot \bigg(I-(1-\I)\sum_{\text{$n$-bit
$p$}}F_n^{PQ} \ketbra{p,0^{n+m}}{p,0^{n+m}}F_n^{PQ}  \bigg) \cdot C^{\dagger QR}\\
= & C^{QR} \cdot \bigg(I-(1-\I)\sum_{\text{$n$-bit
$p$}} \ketbra{p,p,0^{m}}{p,p,0^{m}} \bigg) \cdot C^{\dagger QR}\\
= & I-(1-\I)\sum_{\text{$n$-bit
$p$}} \ketbra{p}{p} \tensor (C^{QR}\ketbra{p,0^{m}}{p,0^{m}} C^{\dagger QR}\\
\end{align*}

Substituting this simplification in $(*)$ above,
\begin{align*}
& \C' \ket{x}\ket{0^n}\ket{0^m} \\
= & \bigg( I-(1-\I)\sum_{\text{$n$-bit
$p$}} \ketbra{p}{p} \tensor (C^{QR}\ketbra{p,0^{m}}{p,0^{m}} C^{\dagger QR}
\bigg)
~~ \ket{x} \bigg(\ket{0}\ket{\psi_0} + \I\ket{1}\ket{\psi_1}\bigg)\\
%= & \ket{x} \bigg(\ket{0}\ket{\psi_0} + \I \ket{1}\ket{\psi_1}\bigg) -\\
%  & \bigg((1-\I)\sum_{\text{$n$-bit
%$p$}} \ketbra{p}{p} \tensor (C^{QR}\ketbra{p,0^{m}}{p,0^{m}} C^{\dagger
%QR}\bigg)
%~~ \ket{x} \bigg(\ket{0}\ket{\psi_0} + \I \ket{1}\ket{\psi_1}\bigg)\\
= & \ket{x} \bigg(\ket{0}\ket{\psi_0} + \I \ket{1}\ket{\psi_1}\bigg) -\\
  & (1-\I)\sum_{\text{$n$-bit
$p$}} \ket{p}\braket{p}{x} \tensor \bigg(C^{QR}\ketbra{p,0^{m}}{p,0^{m}} C^{\dagger
QR} \bigg)
~~ \bigg(\ket{0}\ket{\psi_0} + \I\ket{1}\ket{\psi_1}\bigg) \\
= & \ket{x} \bigg(\ket{0}\ket{\psi_0} + \I \ket{1}\ket{\psi_1}\bigg) -
   (1-\I) \ket{x} \tensor \bigg(C^{QR}\ketbra{x,0^{m}}{x,0^{m}} C^{\dagger
QR} \bigg)
~~ \bigg(\ket{0}\ket{\psi_0} + \I \ket{1}\ket{\psi_1}\bigg) \\
= & \ket{x} \bigg( \big(\ket{0}\ket{\psi_0} + \I \ket{1}\ket{\psi_1}\big) -
(1-\I) \big(\ket{0}\ket{\psi_0} + \ket{1}\ket{\psi_1}\big)
\big(\bra{0}\bra{\psi_0} + \bra{1}\bra{\psi_1}\big) 
~~ \big(\ket{0}\ket{\psi_0} + \I \ket{1}\ket{\psi_1}\big)\bigg)\\
= & \ket{x} \bigg( \big(\ket{0}\ket{\psi_0} + \I \ket{1}\ket{\psi_1}\big) -
(1-\I) \big(\ket{0}\ket{\psi_0} + \ket{1}\ket{\psi_1}\big)
\big( \braket{\psi_0}{\psi_0} + \I\braket{\psi_1}{\psi_1} \big) \bigg)\\
= & \ket{x} \bigg( \big(1 - (1-\I)K\big) \ket{0}\ket{\psi_0} + \big(\I -
(1-\I)K\big) \ket{1}\ket{\psi_1} \bigg)~~\text{where,
}K=\braket{\psi_0}{\psi_0} + \I\braket{\psi_1}{\psi_1}\\
= & \left\{ 
    \begin{array}{ll}
	\I\ket{x}\ket{0}\ket{\psi_0} & \text{ if, } x \not\in L \text{ i.e., }
	\braket{\psi_1}{\psi_1}=0,~\braket{\psi_0}{\psi_0}=1\\
	(\I-1)\ket{x}\ket{1}\ket{\psi_1} & \text{ if, } x \in L \text{ i.e., }
	\braket{\psi_1}{\psi_1}=\braket{\psi_0}{\psi_0}=1/2
    \end{array}
\right.%}
\end{align*}

Measuring the first qubit of register $Q$ therefore shows $\ket{1}$ if and only
if $x \in L$.

\end{proof}

\section{Optimal values for Grover iterator}
Let $c$ denote $\big( e^{\I\theta} + e^{\I\alpha} - 1 +
(1-e^{\I\alpha})(1-e^{\I\theta})p \big)$. 
Then, $c^* = -(1-p) + 2(1-p)e^{-\I\theta} + pe^{-2\I\theta}$.
Therefore, if $p > 0$, then $\Delta = cc^*$ which we will compute below.
\begin{proof}[Computing $\Delta$]
\begin{align*}
       & \Delta = c c^* \\
       & = (1-p)^2 - 2(1-p)^2 e^{-\I\theta} - p(1-p)e^{-2\I\theta}\\
       & - 2(1-p)^2 e^{\I\theta} + 4(1-p)^2 + 2p(1-p)e^{-\I\theta}\\
       & -p(1-p)e^{2\I\theta} + 2p(1-p)e^{\I\theta} + p^2\\
       & = [(1-p)^2 + 4(1-p)^2 + p^2] + (e^{-\I\theta}+e^{\I\theta})[2p(1-p)-2(1-p)^2] 
    - (e^{-2\I\theta}+e^{2\I\theta}) p(1-p) \\
       & = 6p^2 - 10p + 5 + 4(1-p)(2p-1)\cos\theta - 2p\cos 2\theta + 2p^2\cos
    2\theta\\
       & = (-10 - 2\cos 2\theta)p + (6 + 2\cos 2\theta)p^2 + (\sin^2\theta +
    \cos^2\theta) + 4 + 4(1-p)(2p-1)\cos\theta\\
       & = (-8 - 4\cos^2\theta)p + (4+4\cos^2\theta)p^2 + \sin^2\theta +
    \cos^2\theta + 4 + 4(1-p)(2p-1)\cos\theta\\
       & = \sin^2\theta + (4p^2 -4p + 1)\cos^2\theta + 4 + 4p^2 - 8p +
    4(1-p)(2p-1)\cos\theta\\
       & = \sin^2\theta + (2p-1)^2\cos^2\theta + 4(1-p)^2 +
    4(1-p)(2p-1)\cos\theta\\
       & = [(2p-1)\cos\theta + 2(1-p)]^2 + \sin^2\theta
\end{align*}
\end{proof}

\ckinduction*

\begin{proof}
We will give a quick sketch of the proof by induction.

For $k=1$, $C_1 = G^*_\epsilon = C S_{\ket{\psi}} C^\dagger S_\Proj C = \Q C$
so the claim holds for the base case.

Now, suppose that the claim holds for some $1 \le j < k$. Before discussing the
induction case, note that $(\Q^\dagger)^t = (S_\Proj C S_{\ket{\psi}}
C^\dagger)^t = S_\Proj \cdot \Q^{t-1} \cdot (C S_{\ket{\psi}}
C^\dagger)$ for any $t$.

Then, $C_{j+1} = G^*_{\epsilon_j}(C_j,
\ket{\psi}, \Proj) = C_j S_{\ket{\psi}} C_j^\dagger S_\Proj C_j$ which, using
the induction hypothesis, is $\Q^{(3^j-1)/2} C \cdot S_{\ket{\psi}} \cdot
C^\dagger (\Q^{\dagger})^{(3^j-1)/2} \cdot S_\Proj \Q^{(3^j-1)/2} C = $ (using the
expression for $(\Q^\dagger)^t$ above) $\Q^{(3^j-1)/2 + 1 + (3^j-1)/2 - 1 + 1
+ (3^j-1)/2} C = \Q^{(3^{j+1}-1)/2}C$.
\end{proof}

\section{Optimum initial state for distinguishing two circuits}
Recall that $| \braket{\psi_1}{\psi_2} | = |
\braket{\phi|C_1^\dagger C_2}{\phi} |$. Denoting $C_1^\dagger C_2$ by $S$, we
would like to minimize $| \braket{\phi|S}{\phi} |$ over all possible pure
state $\ket{\phi}$. Suppose the eigenvalues of $S$ are $e^{i\theta_1},
\ldots$ with corresponding eigenvectors $\ket{v_1}, \ldots$. Using a recent result~\cite{2015BeraATPG}, the
maximum value of $\epsilon$ is obtained by solving the optimization problem
$$\min f(\theta_1, \ldots) = \bigg(\sum_j c_j^2 + \sum_{j \not= k} c_j c_k \cos(\theta_j -
\theta_k) \bigg), ~~~\mbox{ where, } \sum_j c_j=1,~~~~ 0 \le c_j \le 1$$

Suppose $f_{OPT}$ denotes the optimal value above and $c_1, \ldots $ denote the
corresponding solution. Then, the optimal $\epsilon$ is $1-f_{OPT}^2$ and $\ket{\phi}$ can be set to $\sum_j \sqrt{c_j} \ket{v_j}$.

\end{document}